\documentclass[final]{aomart}
\usepackage[english]{babel}
\usepackage{graphicx}
\usepackage{bbm}
\usepackage{subcaption}
\usepackage{blkarray}
\usepackage{comment}

\usepackage{floatrow}
\makeatletter 
\fancypagestyle{firstpage}{%
  \fancyhf{}%
  \if@aom@manuscript@mode
    \lhead{\begin{picture}(0,0)%
        \put(-21,-25){\usebox{\@aom@linecount}}%
      \end{picture}}
  \fi
  \chead{}%
   \cfoot{\footnotesize\thepage}}%
\doinumber{}
\makeatother

\DeclareMathOperator{\dist}{dist}
\DeclareMathOperator{\stv}{\textsc{STV}}
\DeclareMathOperator{\topp}{\textsf{top}}
\DeclareMathOperator{\distortion}{\textsf{distortion}}
\DeclareMathOperator{\coordination}{\textsc{Coordination}}
\DeclareMathOperator{\pluralitymatching}{\textsc{PluralityMatching}}
\DeclareMathOperator{\socialcost}{\textsc{SC}}

\usepackage[capitalize, nameinlink]{cleveref}

\theoremstyle{plain}
\newtheorem{theorem}{Theorem}[section]
\newtheorem{lemma}[theorem]{Lemma}
\newtheorem{corollary}[theorem]{Corollary}
\newtheorem{proposition}[theorem]{Proposition}

\newtheorem{conjecture}[theorem]{Conjecture}
\newtheorem{observation}[theorem]{Observation}
\newtheorem{claim}[theorem]{Claim}
\newtheorem{assumption}[theorem]{Assumption}

\newtheorem{question}{Question}

\theoremstyle{definition}
\newtheorem{definition}[theorem]{Definition}

\theoremstyle{remark}
\newtheorem{remark}[theorem]{Remark}

\title[Dimensionality, Coordination, and Robustness]{Dimensionality, Coordination, and Robustness in Voting}
\address{Department of Computer Science, National Technical University of Athens}
\email{ioannis.anagnostides@gmail.com}
\email{fotakis@cs.ntua.gr}
\email{patsilinak@corelab.ntua.gr}
\author[I. Anagnostides, D. Fotakis, and P. Patsilinakos]{Ioannis Anagnostides, Dimitris Fotakis, and Panagiotis Patsilinakos}

\copyrightnote{}

\begin{document}

\begin{abstract}
We study the performance of voting mechanisms from a utilitarian standpoint, under the recently introduced framework of metric-distortion, offering new insights along three main lines. First, if $d$ represents the doubling dimension of the metric space, we show that the distortion of $\stv$ is $O(d \log \log m)$, where $m$ represents the number of candidates. For doubling metrics this implies an exponential improvement over the lower bound for general metrics, and as a special case it effectively answers a question left open by Skowron and Elkind (AAAI `17) regarding the distortion of $\stv$ under low-dimensional Euclidean spaces. More broadly, this constitutes the first nexus between the performance of any voting rule and the ``intrinsic dimensionality'' of the underlying metric space. We also establish a nearly-matching lower bound, refining the construction of Skowron and Elkind. Moreover, motivated by the efficiency of $\stv$, we investigate whether natural learning rules can lead to low-distortion outcomes. Specifically, we introduce simple, deterministic and decentralized exploration/exploitation dynamics, and we show that they converge to a candidate with $O(1)$ distortion. Finally, driven by applications in facility location games, we consider several refinements and extensions of the standard metric setting. Namely, we prove that the deterministic mechanism recently introduced by Gkatzelis, Halpern, and Shah (FOCS `20) attains the optimal distortion bound of $2$ under ultra-metrics, while it also comes close to our lower bound under distances satisfying approximate triangle inequalities.  
\end{abstract}

\maketitle
\clearpage
\tableofcontents

\section{Introduction}

Aggregating the preferences of individual entities into a collective decision lies at the foundations of voting theory, and has recently found a myriad of applications in areas such as information retrieval, recommender systems, and machine learning \cite{DBLP:journals/jmlr/LuB14}. A common hypothesis in the literature of social choice asserts that agents only provide an \emph{order} of preferences over a (finite) set of alternatives, without indicating a precise \emph{measure} of each preference. However, this assertion might seem misaligned with many classical models in economic theory \cite{vonNeumann1944-VONTOG-4} which espouse a \emph{utilitarian} framework to represent agents' preferences. This raises the following concern: What is the loss in utilitarian efficiency of a mechanism eliciting only ordinal information? 

This question was raised by Procaccia and Rosenschein \cite{DBLP:conf/cia/ProcacciaR06}, introducing the concept of \emph{distortion}, and has since led to a substantial body of work. In this paper we mostly focus on the refined notion of \emph{metric distortion} \cite{DBLP:conf/aaai/AnshelevichBP15}, wherein agents and candidates are associated with points in some \emph{metric space}, and preferences are being determined based on the proximity in the underlying metric (see \Cref{section:preliminaries} for a formal definition). Importantly, this framework offers a quantitative ``benchmark'' for comparing different voting rules commonly employed in practice. Indeed, one of the primary considerations of our work lies in characterizing the performance of the \emph{single transferable vote}\footnote{For consistency with prior work $\stv$ will represent throughout this paper the \emph{single-winner} variant of the system, which is sometimes referred to as \emph{instant-runoff voting} (\textsc{IRV}) in the literature.} mechanism (henceforth $\stv$). 

$\stv$ is a widely-popular iterative voting system employed in the national elections of several countries, including Australia, Ireland, and India, as well as in many other preference aggregation tasks; e.g., in the Academy Awards. To be more precise, $\stv$ proceeds in an iterative fashion: In each round, agents vote for their most preferred candidate---among the \emph{active} ones, while the candidate who enjoyed the least amount of support in the current round gets eliminated. This process is repeated for $m-1$ rounds, where $m$ represents the number of (initial) alternatives, and the last surviving candidate is declared the winner of $\stv$. As an aside, notice that this process is generally non-deterministic due to the need for a tie-breaking mechanism; as in \cite{DBLP:conf/aaai/SkowronE17}, we will work with the \emph{parallel universe model} of Conitzer et al. \cite{10.5555/1661445.1661464}, wherein a candidate is said to be an $\stv$ winner if it survives under \emph{some} sequence of eliminations. 

In this context, Skowron and Elkind \cite{DBLP:conf/aaai/SkowronE17} were the first to analyze the distortion of $\stv$ under metric preferences. Specifically, they showed that the distortion of $\stv$ in general metric spaces is always $O(\log m)$, while they also gave a nearly-matching lower bound in the form of $\Omega(\sqrt{\log m})$. Interestingly, a careful examination of their lower bound reveals the existence of a high-dimensional submetric, as depicted in \Cref{fig:high_dimensionall}, and it is a well-known fact in the theory of metric embeddings that such objects cannot be \emph{isometrically} embedded into low-dimensional\footnote{We say that a Euclidean space is low-dimensional if its dimension $d$ is bounded by a ``small'' universal constant, i.e. $d = O(1)$.} Euclidean spaces \cite{Matousek2002}. As a result, Skowron and Elkind \cite{DBLP:conf/aaai/SkowronE17} left open the following intriguing question:

\begin{question}
    \label{question:STV}
What is the distortion of $\stv$ under low-dimensional Euclidean spaces?
\end{question}

\begin{figure}[!ht]
    \centering
    \includegraphics[scale=0.5]{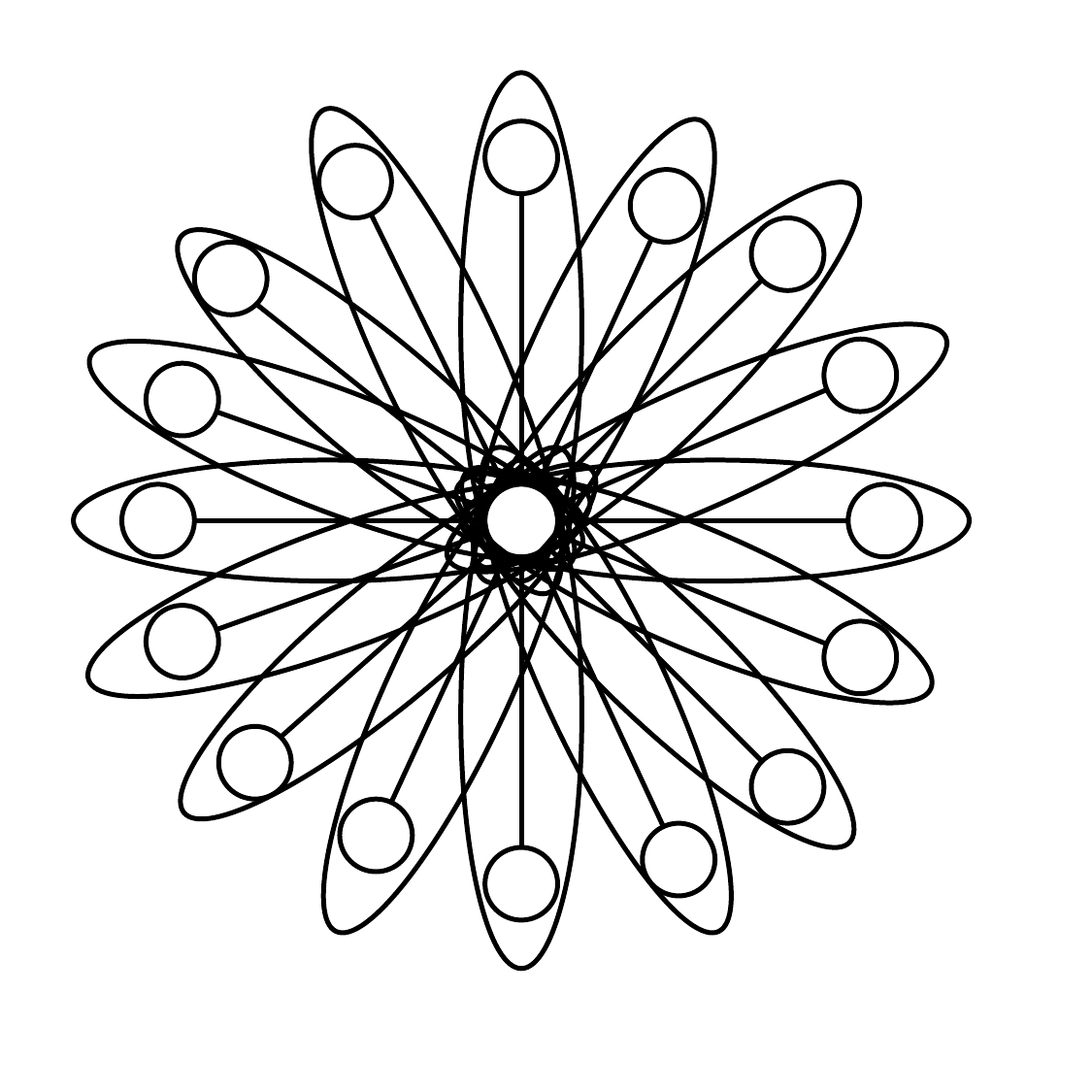}
    \caption{A high-dimensional metric in the form of a ``star'' graph.}
    \label{fig:high_dimensionall}
\end{figure}

Needless to say that the performance of voting rules in low-dimensional spaces has been a subject of intense scrutiny in spatial voting theory, under the premise that voters and candidates are typically embedded in subspaces with small dimension \cite{arrow_1990,RePEc:cup:cbooks:9780521275156}. For example, recent experimental work by Elkind et al. \cite{DBLP:conf/aaai/ElkindFLSST17} evaluates several voting rules in a $2$-dimensional Euclidean space, motivated by the fact that preferences are typically crystallized on the basis of a few crucial dimensions; e.g., economic policy and healthcare. Indeed, in the so-called Nolan Chart---a celebrated political spectrum diagram---political views are charted along two axes, expanding upon the traditional one-dimensional representation; to quote from the work of Elkind et al. \cite{DBLP:conf/aaai/ElkindFLSST17}:

\begin{quote}
    \textit{``...the popularity of the Nolan
Chart [...] indicates that two dimensions
are often sufficient to provide a good approximation of voters’ preferences.''}
\end{quote}

Thus, it is natural to ask whether we can refine the analysis of $\stv$ under low-dimensional spaces. In fact, as part of a broader agenda analogous questions can be raised for other mechanisms as well. However, it is interesting to point out that for many voting rules analyzed within the framework of distortion there exist low-dimensional lower bounds; some notable examples are given in \Cref{tab:dimensions}. In contrast, our work will separate $\stv$ from the mechanisms in \Cref{tab:dimensions}, effectively addressing \Cref{question:STV}. Importantly, we shall provide a characterization well-beyond Euclidean spaces, to metrics with ``intrinsically'' low dimension. 

\begin{table}[!ht]
    \centering
    \begin{tabular}{||c c c||} 
 \hline
 Mechanism & Lower Bound & Dimension \\ [0.5ex] 
 \hline\hline
 Plurality & $2m-1$ & $1$ \\ 
 \hline
 Borda & $2m-1$ & $1$ \\
 \hline
 Copeland & $5$ & $2$ \\
 \hline
 Veto & $2n-1$ & $1$ \\
 \hline
 Approval & $2n-1$ & $1$ \\
 \hline
\end{tabular}
\caption{The Euclidean dimension required to construct a (tight) lower bound for several common voting rules; these results appear in \cite{DBLP:conf/aaai/AnshelevichBP15}. We should note that for Copeland the metric constructed in \cite{DBLP:conf/aaai/AnshelevichBP15} is not Euclidean, but can be easily modified to be one.}
    \label{tab:dimensions}
\end{table}

The next consideration of our work is directly motivated by the efficiency of $\stv$ compared to the plurality rule, and in particular the strategic implications of this discrepancy. A good starting point for this discussion stems from the fact that in many fundamental preference aggregation settings alternatives are chosen by inefficient mechanisms, and in many cases any reform faces insurmountable impediments. For example, in political elections the voting mechanism is typically dictated by electoral laws, or even the constitution \cite{doi:10.1177/0951692892004002005}. As a result, understanding the behavior of strategic agents when faced with inefficient mechanisms is of paramount importance \cite{DBLP:conf/aaai/BrillC15,DBLP:conf/wine/ZuckermanFCR11}. A rather orthogonal way of viewing this is whether autonomous agents can converge to admissible social choices through natural learning rules; this begs the question:

\begin{question}
    \label{question:dynamics}
    To what extent can strategic behavior improve efficiency in voting?
\end{question}

We stress that although in the absence of any information it might be unclear how agents can engage in strategic behavior, in most applications of interest agents have plenty of prior information before they cast their votes, e.g. through polls, surveys, forecasts, prior elections, or even early voting. Indeed, there is a prolific line of work which studies population dynamics for agents that cast their votes in response to the information they possess (see \cite{10.2307/40270930}, and references therein), as well as the role of information in shaping public policy \cite{RePEc:aub:autbar:579.03}.

To address such considerations we propose a natural model wherein agents act iteratively based on some partial feedback on the other voters' preferences. We explain how $\stv$ can be very naturally cast in this framework, while we establish the existence of simple and decentralized \emph{coordination} dynamics converging to a near-optimal alternative. 

The final theme of our work offers certain refinements and extensions of prior works, mostly driven by some fundamental applications in the context of facility location games. Specifically, we primarily focus on the optimal---under metric preferences---deterministic mechanism recently introduced by Gkatzelis, Halpern, and Shah \cite{DBLP:conf/focs/Gkatzelis0020}; we show that it recovers the optimal bound under ultra-metrics, and near-optimal distortion under distances satisfying approximate triangle inequalities.

\subsection{Overview of Results} Our first contribution is to relate the distortion of $\stv$ to the dimensionality of the underlying metric space. Specifically, our first insight is to employ the following fundamental concept from metric geometry:\footnote{To keep the exposition reasonably smooth, the formal definition of standard notation is deferred to the preliminaries in \Cref{section:preliminaries}.}

\begin{definition}[Doubling Dimension]
    \label{definition:doubling_dimension}
The \emph{doubling constant} of a metric space $(\mathcal{M}, \dist)$ is the least integer $\lambda \geq 1$ such that for all $x \in \mathcal{M}$ and for all $r > 0$, every ball $\mathcal{B}(x, 2r)$ can be covered by the union of at most $\lambda$ balls of the form $\mathcal{B}(s, r)$, where $s \in \mathcal{M}$; that is, there exists a subset $\mathcal{S} \subseteq \mathcal{M}$ with $|\mathcal{S}| \leq \lambda$ such that 

\begin{equation}
    \mathcal{B}(x, 2r) \subseteq \bigcup_{s \in \mathcal{S}} \mathcal{B}(s, r).
\end{equation}

The \emph{doubling dimension} is then defined as $\dim(\mathcal{M}) :=  \log_2 \lambda$.\footnote{To avoid trivialities it will be assumed that $\lambda \geq 2$.}
\end{definition}

This concept generalizes the standard notion of dimension since $\dim(\mathbb{R}^d) = \Theta(d)$ when $\mathbb{R}^d$ is endowed with the $\ell_p$ norm. Moreover, it is clear that for a finite metric space $(\mathcal{M}, \dist), \dim(\mathcal{M}) \leq \log_2 |\mathcal{M}|$; for example, this is essentially tight for the high-dimensional metric of \Cref{fig:high_dimensionall}. The concept of doubling dimension was introduced by Larman \cite{10.1112/plms/s3-17.1.178} and Assouad \cite{BSMF_1983__111__429_0}, and was first used in algorithm design by Clarkson \cite{DBLP:conf/stoc/Clarkson97} in the context of the nearest neighbors problem. Nevertheless, we are not aware of any prior characterization that leverages the doubling dimension in the realm of voting theory. In the sequel, it will be assumed that $(\mathcal{M}, \dist(\cdot, \cdot))$ stands for the metric space induced by the set of candidates and voters. In this context, our first main contribution is the following theorem:

\begin{theorem}
    If $d$ is the doubling dimension of $\mathcal{M}$, then the distortion of $\stv$ is $O(d \log \log m)$.
\end{theorem}

For \emph{doubling metrics}\footnote{A doubling metric refers to a metric space with doubling dimension upper-bounded by some universal constant.} this theorem already implies an exponential improvement in the distortion over the $\Omega(\sqrt{\log m})$ lower bound for general metrics. Moreover, it addresses as a special case \Cref{question:STV}:

\begin{corollary}
The distortion of $\stv$ under low-dimensional Euclidean spaces is $O(\log \log m)$.
\end{corollary}

To the best of our knowledge, this is the first result that relates the performance of any voting rule to the ``intrinsic dimensionality'' of the underlying metric space. It also corroborates the experimental findings of Elkind et al. \cite{DBLP:conf/aaai/ElkindFLSST17} regarding the superiority of $\stv$ on the $2$-dimensional Euclidean plane. More broadly, we suspect that our characterization applies for a wide range of iterative voting rules, to which $\stv$ serves as a canonical example. We should note that the $O(\log \log m)$ factor appears to be an artifact of our analysis. Indeed, we put forward the following conjecture:

\begin{conjecture}
If $d$ is the doubling dimension of $\mathcal{M}$, then the distortion of $\stv$ is $O(d)$.
\end{conjecture}

Verifying this conjecture in light of our result might be of small practical importance, but nonetheless we believe that it can be established by extending our techniques. In fact, for one-dimensional spaces we actually confirm this conjecture, proving that the distortion of $\stv$ on the line is $O(1)$ in \Cref{theorem:STV-line}. It should be noted, however, that the underlying phenomenon is inherently different once we turn our attention to higher-dimensional spaces. In addition, to complement our positive results we refine the lower bound of Skowron and Elkind \cite{DBLP:conf/aaai/SkowronE17}, showing an $\Omega(\sqrt{d})$ lower bound, where $d$ represents the doubling dimension of the submetric induced by the set of candidates $\mathcal{M}_C$. Thus, it should be noted that there are still small gaps left to be bridged in future research.

\textit{Other Notions of Dimension}. An important advantage of the doubling dimension is that it essentially subsumes other commonly-used notions of dimension. Most notably, Karger and Ruhl \cite{DBLP:conf/stoc/KargerR02} have introduced a concept of dimension based on the \emph{growth rate} of a (finite) metric space, and it is known (\cite[Proposition 1.2]{DBLP:conf/focs/GuptaKL03}) that the doubling dimension can only be a factor of $4$ larger than the growth rate of Karger and Ruhl. Moreover, a similar statement applies for the \emph{local density} of an unweighted graph, another natural notion of volume that has been employed in the analysis of a graph's bandwidth \cite{FEIGE2000510}.

\textit{High-Level Intuition.} In this paragraph we briefly attempt to explain why the distortion of $\stv$ depends on the ``covering dimension'' of the underlying metric space. First, we have to describe the technique developed by Skowron and Elkind \cite{DBLP:conf/aaai/SkowronE17}. Specifically, their method for deriving an upper bound for the distortion of an iterative voting rule consists of letting a substantial fraction of agents reside within close proximity to the optimal candidate, and then analyze how the support of these agents propagates throughout the evolution of the iterative process. More precisely, the overall distance covered immediately implies an upper bound on the distortion (see \Cref{lemma:wave}). The important observation is that the underlying dimension drastically affects this phenomenon. In particular, when a large fraction of agents lies in a low-dimensional ball supporting many different candidates, we can infer that their (currently) second most-preferred alternatives ought to be ``close''---for most of the agents---by a covering argument (and the triangle inequality). This directly circumscribes the propagation of the support, as hinted in \Cref{fig:sub2}, juxtaposed to the phenomenon in high dimensions in \Cref{fig:sub1}. We stress that we shall make use of this basic skeleton developed by Skowron and Elkind \cite{DBLP:conf/aaai/SkowronE17}. We should also remark that we prove the $O(\log m)$ bound under general metrics through a simpler analysis (see \Cref{theorem:STV-general_metrics}), which incidentally reveals a very clean recursive structure; this argument will be directly invoked for the proof of our main theorem.

\begin{figure}[!ht]
\centering
\begin{subfigure}{.5\textwidth}
  \centering
  \includegraphics[scale=0.43]{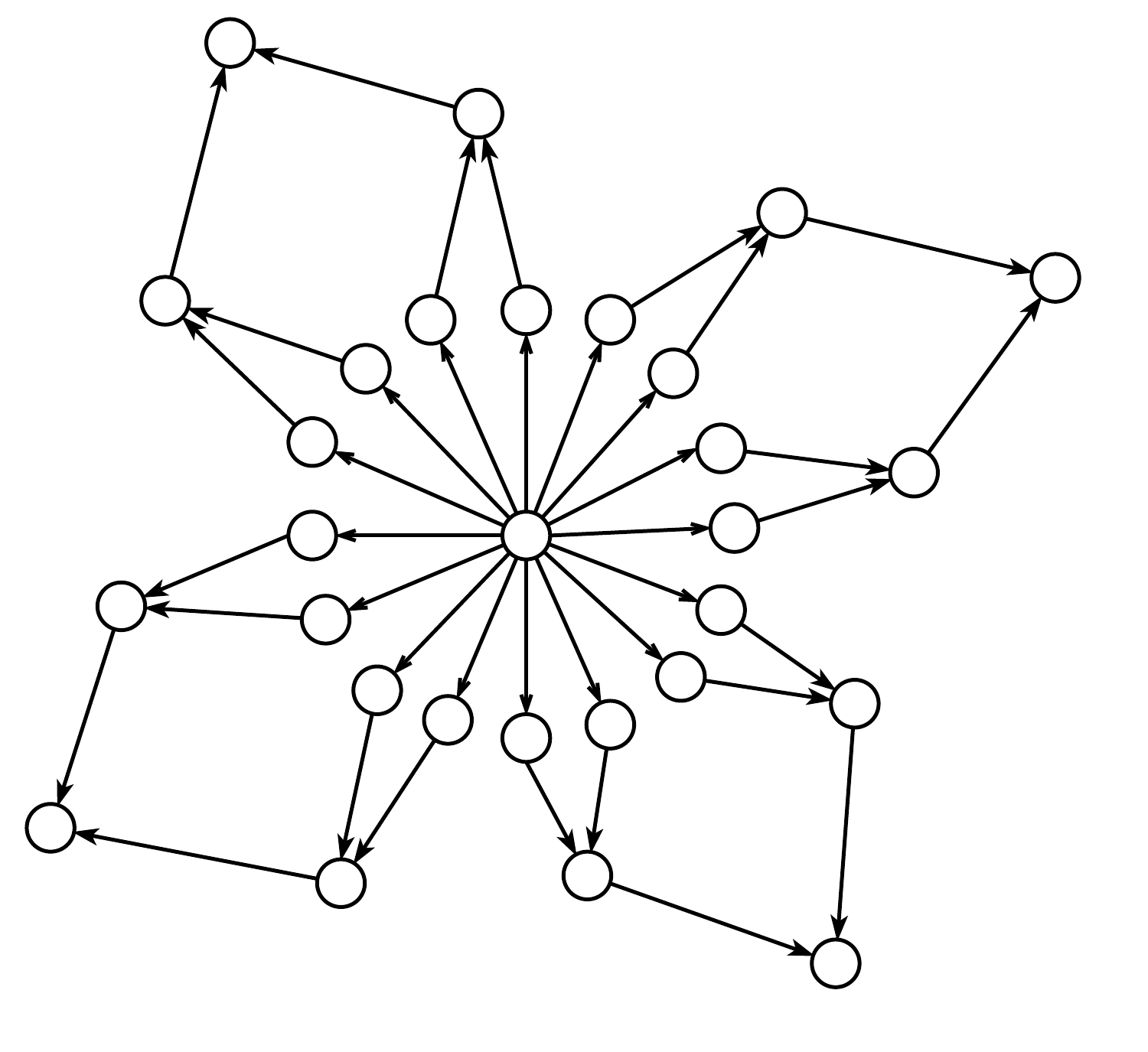}
  \caption{The propagation of the support in \\high dimensions.}
  \label{fig:sub1}
\end{subfigure}%
\begin{subfigure}{.5\textwidth}
  \centering
  \includegraphics[scale=0.35]{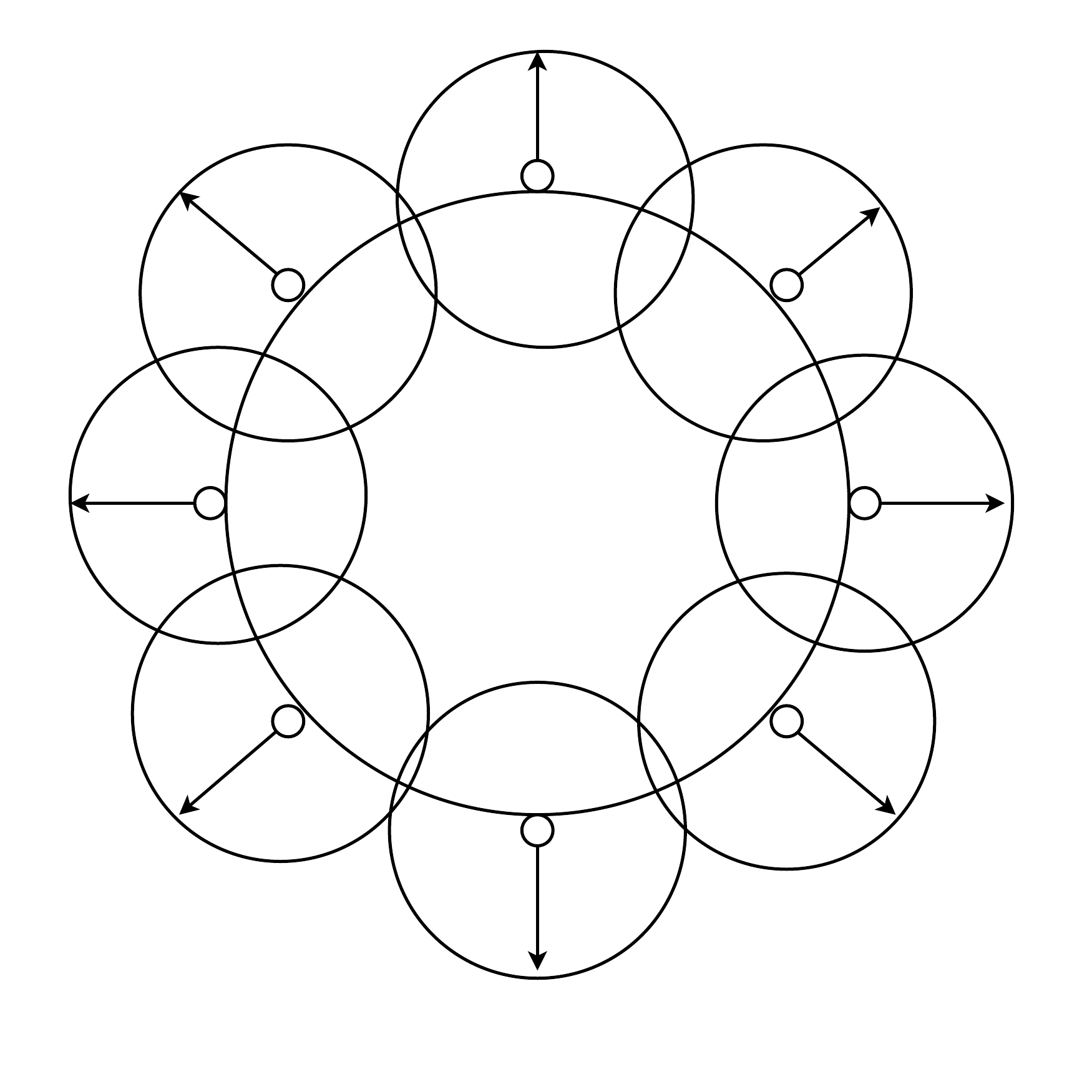}
  \caption{The propagation of the support in \\low dimensions.}
  \label{fig:sub2}
\end{subfigure}
\caption{The impact of the underlying dimension on $\stv$.}
\label{fig:test}
\end{figure}

The next theme of our work is motivated by the performance of $\stv$, and in particular offers a preliminary answer to \Cref{question:dynamics}. Specifically, to formally address such questions we first propose a natural iterative model: In each day every agent has to select a \emph{single} candidate, and at the end of the round agents are informed about the (plurality) scores of all the candidates (cf., see \cite{BorodinL0S19}). This process is repeated for sufficiently many days, and it is assumed that the candidate who enjoyed the largest amount of support in the ultimate day will eventually prevail. Observe that in this scenario truthful engagement appears to be very unrealistic since agents would endeavor to adapt their support based on the popularity of each candidate; for example, it would make little sense to squander one's vote (at least towards the last stages) to an unpopular candidate. More broadly, there is an interesting nexus between distortion and stability, as we elaborate in \Cref{section:coordination}, emphasizing on a connection with the notion of \emph{core} in cooperative game theory (\Cref{proposition:core}). 

In this context, $\stv$ already suggests a particularly natural strategic engagement, improving exponentially over the outcome of the truthful dynamics. Yet, it yields super-constant distortion due to the greedy aspect of the induced dynamics. We address this issue by designing a simple and decentralized exploration/exploitation scheme: 

\begin{theorem}
    There exist simple, deterministic and distributed dynamics that converge to a candidate with $O(1)$ distortion.
\end{theorem}

We elaborate on the proposed dynamics, as well as on all the aforementioned issues in \Cref{section:coordination}.

The final contribution of our work concerns refinements and extensions of prior results under metric preferences, providing new insights along two main lines. First, we study preference aggregation under ordinal information when agents and candidates are located in \emph{ultra-metric spaces}, which is a strengthening of the standard metric assumption. This setting is mostly motivated by the fundamental \emph{bottleneck} variant in facility location games, wherein the cost of a path between an agent and a server corresponds to the largest weight among the edges in the path, instead of their sum \cite{GABOW1988411}; it should be noted that ultra-metrics also commonly arise in branches of mathematics such as metric geometry \cite{ABRAHAM20113026}, and $p$-adic analysis \cite{rankin_1966}. In this context, our main observation is that the $\pluralitymatching$ mechanism of Gkatzelis, Halpern, and Shah \cite{DBLP:conf/focs/Gkatzelis0020} always obtains distortion $2$, which incidentally is the provable lower bound for any deterministic mechanism. It is particularly interesting that the optimal mechanism under metric spaces retains its optimality under an important refinement, illustrating the robustness of $\pluralitymatching$.

We also study the performance of $\pluralitymatching$ under distance functions that satisfy a $\rho$\emph{-relaxed} triangle inequality. This consideration is directly driven by the fact that many well-studied and commonly-arising distances are only approximately metrics (most notably, the \emph{squared Euclidean} distance is a $2$-approximate metric), but we believe that there is another concrete reason. Most research in the realm of distortion has thus far been divided between the metric case and the \emph{unit-sum} case, with these two lines of research being largely disconnected. Studying approximate metrics serves as an attempt to bridge this gap. In this context, we prove a lower bound of $\rho^2 + \rho + 1$, while $\pluralitymatching$ incurs distortion at most $2\rho^2 + \rho$, thus leaving a small gap for future research. Notice that for the special case $\rho := 1$ this recovers the result of Gkatzelis, Halpern, and Shah \cite{DBLP:conf/focs/Gkatzelis0020}.

\subsection{Related Work}

The framework of distortion under \emph{metric preferences} was first introduced a few years ago by Anshelevich et al. \cite{DBLP:conf/aaai/AnshelevichBP15} (see also \cite{DBLP:journals/ai/AnshelevichBEPS18}). Specifically, they observed a lower bound of $3$ for any deterministic mechanism, while they also showed---among others---that Copeland's method, a very popular voting system, always incurs distortion at most $5$, with the bound being tight for certain instances. This threshold was subsequently improved by Munagala and Wang \cite{DBLP:conf/ec/MunagalaW19}, introducing a novel (deterministic) mechanism with distortion $2 + \sqrt{5}$, while the same bound was independently obtained by Kempe \cite{DBLP:conf/aaai/000120a} through an approach based on LP duality. The lower bound of $3$ was only recently matched by $\pluralitymatching$, a mechanism introduced by Gkatzelis, Halpern, and Shah \cite{DBLP:conf/focs/Gkatzelis0020}. In \Cref{section:robustness} we investigate the performance of this mechanism under certain refinements and extensions, leveraging an important property established in \cite{DBLP:conf/focs/Gkatzelis0020} regarding the existence of a \emph{perfect fractional matching} on a certain bipartite graph. 

All of the aforementioned results apply under arbitrary metric spaces. Several special cases have also attracted attention in the literature. For one-dimensional spaces, Feldman et al. \cite{DBLP:conf/sigecom/FeldmanFG16} establish several improved bounds, while a comprehensive characterization in a \emph{distributed} setting was recently given by Filos-Ratsikas and Voudouris in \cite{DBLP:journals/corr/abs-2007-06304}. Another notable refinement germane to our considerations in \Cref{section:robustness} was studied by Anshelevich and Postl \cite{DBLP:conf/ijcai/AnshelevichP16} in the form of $\alpha$-\emph{decisiveness}, imposing that voters support their top choices by a non-negligible margin. This condition has led to several refined upper and lower bounds; cf. see \cite{DBLP:conf/focs/Gkatzelis0020}. The interested reader is referred to the concise survey of Anshelevich et al. \cite{anshelevich2021distortion} for detailed accounts on the rapidly growing literature on the subject. Moreover, for related research beyond the framework of distortion we refer to~\cite{Gershkov19:Voting}, and references therein.

The model we introduce in \Cref{section:coordination} is related to the seminal work of Branzei, Caragiannis, Morgenstern, and Procaccia \cite{DBLP:conf/aaai/BranzeiCMP13} (see also the extensive follow-up work, such as \cite{DBLP:conf/aaai/ObraztsovaMPRJ15}), viewing voting from the standpoint of \emph{price of anarchy} (PoA). In particular, the authors study the discrepancy between the plurality scores under truthfulness, and under worst-case limit points of \emph{best-response} dynamics. Instead, we argue that the utilitarian performance of a voting rule---in terms of distortion---offers a very compelling alternative to study this discrepancy, similarly to the original formulation of PoA in the context of routing games \cite{DBLP:conf/stacs/KoutsoupiasP99}, while going beyond best-response dynamics is very much in line with the modern approach in the context of learning in games \cite{DBLP:books/daglib/0016248}. Finally, we stress that \Cref{question:dynamics} has already received extensive attention in the literature (cf. see \cite{DBLP:conf/aaai/BrillC15,DBLP:conf/wine/ZuckermanFCR11} and references therein), but it was not addressed within the framework of (metric) distortion.

\section{Preliminaries}
\label{section:preliminaries}

A \emph{metric space} is a pair $(\mathcal{M}, \dist(\cdot, \cdot))$, where $\dist: \mathcal{M} \times \mathcal{M} \mapsto \mathbb{R}$ is a \emph{metric} on $\mathcal{M}$, i.e., (i) $\forall x,y \in \mathcal{M}, \dist(x, y) = 0 \iff x = y$ (identity of indiscernibles), (ii) $\forall x, y \in \mathcal{M}, \dist(x,y) = \dist(y, x)$ (symmetry), and (iii) $\forall x, y, z \in \mathcal{M}, \dist(x, y) \leq \dist(x, z) + \dist(z, y)$ (triangle inequality). Now consider a set of $n$ voters $V = \{1, 2, \dots, n \}$, and a set of $m$ candidates $C$; we will reference candidates with lowercase letters such as $a, b, w, x$. Voters and candidates are associated with points in a finite metric space $(\mathcal{M}, \dist)$, while it is assumed that $\mathcal{M}$ is the (finite) set induced by the set of voters and candidates. The goal is to select a candidate $x$ who minimizes the \emph{social cost}: $\socialcost(x) = \sum_{i=1}^{n} \dist(i, x)$. This task would be trivial if we had access to the agents' distances from all the candidates. However, in the \emph{metric distortion} framework every agent $i$ provides only a \emph{ranking} (a total order) $\sigma_i$ over the points in $C$ according to the \emph{order} of $i$'s distances from the candidates, with ties broken arbitrarily. We also define $\sigma:= (\sigma_1, \dots, \sigma_n)$, while we will sometimes use $\topp(i)$ to represent $i$'s most preferred alternative.

A \emph{deterministic} \emph{social choice rule} is a function that  maps an \emph{election} in the form of a $3$-tuple $\mathcal{E} = (V, C, \sigma)$ to a single candidate $a \in C$. We will measure the performance of $f$ for a given input of preferences $\sigma$ in terms of its \emph{distortion}; namely, the worst-case approximation ratio it provides with respect to the social cost:

\begin{equation}
    \distortion(f; \sigma) = \sup \frac{\socialcost(f(\sigma))}{\min_{a \in C} \socialcost(a)},
\end{equation}
where the supremum is taken over all metrics consistent with the voting profile. The distortion of a social choice rule $f$ is the maximum of $\distortion(f; \sigma)$ over all possible input preferences $\sigma$. To put it differently, once the mechanism selects a candidate (or a distribution over candidates if the social choice rule is \emph{randomized}) an adversary can select any metric space subject to being consistent with the input preferences. These definitions naturally apply for refinements and extensions studied in the present work.

We define the \emph{open ball} on the metric space $(\mathcal{M}, \dist)$ with center $x \in \mathcal{M}$ and radius $r > 0$ as $\mathcal{B}(x, r) := \{z \in \mathcal{M} : \dist(z, x) < r \}$. An alternative definition for the doubling dimension considers the diameter of subsets, instead of the radius of balls; that is, the doubling constant is the smallest value of $\lambda$ such that every subset of $\mathcal{M}$ can be covered by at most $\lambda$ subsets of (at most) half the diameter. According to this definition, for any submetric $\mathcal{X} \subseteq \mathcal{M}$ it follows that $\dim(\mathcal{X}) \leq \dim(\mathcal{M})$. Nonetheless, it will be convenient to work with the initial notion (\Cref{definition:doubling_dimension}) since switching between the two definitions can only affect the dimension by at most a factor of $2$ (see \cite{DBLP:conf/focs/GuptaKL03}). The following standard covering lemma will be useful for the analysis of $\stv$ in doubling metrics.

\begin{lemma}
\label{lemma:covering}
Consider a metric space $(\mathcal{M}, \dist)$ with doubling constant $\lambda \geq 1$. Then, for any $x \in \mathcal{M}$ and $r > 0$, the ball $\mathcal{B}(x, r)$ can be covered by at most $\lambda^{\lceil \log(r/\epsilon) \rceil}$ balls of radius at most $\epsilon$. 
\end{lemma}

\begin{proof}
If we apply apply \Cref{definition:doubling_dimension} successively we can conclude that any ball of radius $r$ can be covered by at most $\lambda^i$ balls of radius $r/2^i$. Thus, taking $i := \lceil \log(r/\epsilon) \rceil$ leads to the desired conclusion.
\end{proof}

It should be noted that (when unspecified) the $\log(\cdot)$ will always be implied to the base $2$. We conclude this section with a useful lemma observed by Skowron and Elkind \cite{DBLP:conf/aaai/SkowronE17}, which will be used for analyzing iterative voting rules.

\begin{lemma}[\cite{DBLP:conf/aaai/SkowronE17}]
    \label{lemma:wave}
    Consider two distinct candidates $a, b \in C$. If $r := \dist(a, b)/h$ for some parameter $h > 0$, and at most $\gamma n$ agents reside in $\mathcal{B}(a, r)$ for some $\gamma \in [0, 1)$, then 
    
    \begin{equation}
        \frac{\socialcost(b)}{\socialcost(a)} \leq 1 + \frac{h}{1 - \gamma}.
    \end{equation}
\end{lemma}

\begin{proof}
The triangle inequality implies that
\begin{align*}
    \frac{\socialcost(b)}{\socialcost(a)} = \frac{\sum_{i \in V} \dist(i, b)}{\sum_{i \in V} \dist(i, a)} &\leq \frac{\sum_{i \in V} (\dist(i, a) + \dist(a, b))}{\sum_{i \in V} \dist(i, a)} \\
    &= 1 + n \frac{\dist(a,b)}{\sum_{i \in V} \dist(i, a)} \\
    &\leq 1 + \frac{\dist(a,b)}{(1 - \gamma) r} \\
    &= 1 + \frac{h}{1 - \gamma}.
\end{align*}
\end{proof}

\section{STV in Doubling Metrics}

\subsection{STV on the Line} As a warm-up, we will analyze the performance of $\stv$ on the line. In particular, the purpose of this subsection is to establish the following result:

\begin{theorem}
    \label{theorem:STV-line}
    The distortion of $\stv$ on the line is at most $15$.
\end{theorem}

Before we proceed with the proof of this theorem a few remarks are in order. First of all, we did not pursue optimizing the constant in the theorem, although this might be an interesting avenue for future research. It should also be noted that \Cref{theorem:STV-line} already implies a stark separation between $\stv$ and $\textsc{Plurality}$, as the latter is known to admit a one-dimensional $\Omega(m)$ lower bound (recall \Cref{tab:dimensions}).

\begin{proof}[Proof of \Cref{theorem:STV-line}]
Let $w \in C$ be the winner of $\stv$ under some (fixed) sequence of eliminations, and $x \in C$ be the candidate who minimizes the social cost. In the sequel it will be assumed that $\dist(x, w) > 0$; in the contrary case the theorem follows trivially. Moreover, we let $r := d(x, w)/7$, and we consider a sequence of balls $\{ \mathcal{B}_i \}_{i=1}^4$ so that every ball $\mathcal{B}_i$ has center at $x$ and radius $(2i - 1) \times r$, for $i=1,2,3,4$. We will show that at most half of the voters could reside in $\mathcal{B}_1$.

For the sake of contradiction, let us assume that at least a $\gamma > 1/2$ fraction of the voters are in $\mathcal{B}_1$; that is, $ \sum_{i \in V} \mathbbm{1} \{ i \in \mathcal{B}_1 \} = \gamma n > n/2$. First, we will argue that at the time the last candidate in $\mathcal{B}_i$ gets eliminated there is always a candidate located in $\mathcal{B}_{i+1} \setminus \mathcal{B}_i$. Indeed, in the contrary case we can deduce that the last candidate to be eliminated from $\mathcal{B}_i$ would receive the support of all the voters in $\mathcal{B}_1$, which is a contradiction since by construction all the candidates in $\mathcal{B}_4$ have to be eliminated (this follows given that $w \notin \mathcal{B}_4)$. Now consider the stage of $\stv$ just before the last candidate from $\mathcal{B}_1$ was eliminated; observe that this is well-defined as $\mathcal{B}_1$ initially contains at least one candidate, namely $x \in C$. Let us denote with $a^{\ell}, a^{r} \in C$ the leftmost and the rightmost (respectively) nearest \emph{active} candidates from $\mathcal{B}_1$---which are not in $\mathcal{B}_1$. We shall distinguish between two cases:

\textbf{Case I.} $ a^{\ell} \in \mathcal{B}_2 \setminus \mathcal{B}_1$ and $a^r \in \mathcal{B}_2 \setminus \mathcal{B}_1$. Following the elimination of the last candidate from $\mathcal{B}_1$ every voter in $\mathcal{B}_1$ will support either $a^{\ell}$ or $a^r$. Thus, (by the pigeonhole principle) we can conclude that one of these two candidates accumulates at least $n/4$ supporters in the round following the elimination of the last candidate from $\mathcal{B}_1$. Let us assume without any loss of generality that this candidate is $a^{\ell}$. It is important to point out that the support of a candidate can only grow throughout the execution of $\stv$, until elimination. Now consider the stage just before $a^{\ell}$ gets eliminated. We previously argued that there exists a candidate $y \in \mathcal{B}_4 \setminus \mathcal{B}_3$ that has to remain active after the elimination of $a^{\ell}$. This implies that $y$ received at least $n/4$ votes at this particular stage---just before the elimination of $a^{\ell}$. However, observe that none of the supporters of $y$ can derive from $\mathcal{B}_1$ since for every voter in $\mathcal{B}_1$ candidate $a^{\ell}$ is (strictly) superior to $y$. Moreover, the eventual winner $w$ should also have at least $n/4$ supporters at this stage in order to qualify for the next round, but the supporters of $w$ are certainly not from $\mathcal{B}_1$, and are also different from the supporters of $y$. This follows since both are active at this stage and $y \neq w$, implied by the fact that $d(x, y) < 7 r \leq d(x, w)$. As a result, we have concluded that the total number of voters is strictly more than $n$, which is an obvious contradiction.

\textbf{Case II.} Only one of the candidates $a^{\ell}$ and $a^r$ resides in $\mathcal{B}_2 \setminus \mathcal{B}_1$. Notice that our previous argument implies that at least one of the two should be in $\mathcal{B}_2 \setminus \mathcal{B}_1$, and hence, there is indeed no other case to consider. Let us assume without any loss of generality that the candidate in $\mathcal{B}_2 \setminus \mathcal{B}_1$ is $a^{\ell}$. Observe that just before the last candidate from $\mathcal{B}_1$ gets eliminated every voter in $\mathcal{B}_1$ either supports the candidate in $\mathcal{B}_1$ or candidate $a^{\ell}$. Moreover, we have argued that there are at least $5$ remaining candidates. Thus, we can infer that the last candidate from $\mathcal{B}_1$ received at most $20 \%$ of the votes at the round of elimination, in turn implying that $a^{\ell}$ enjoyed the support of at least $30\%$ of the voters. Consequently, we can easily reach a contradiction similarly to the previously considered case.

\begin{figure}[!ht]
    \centering
    \includegraphics[scale=0.46]{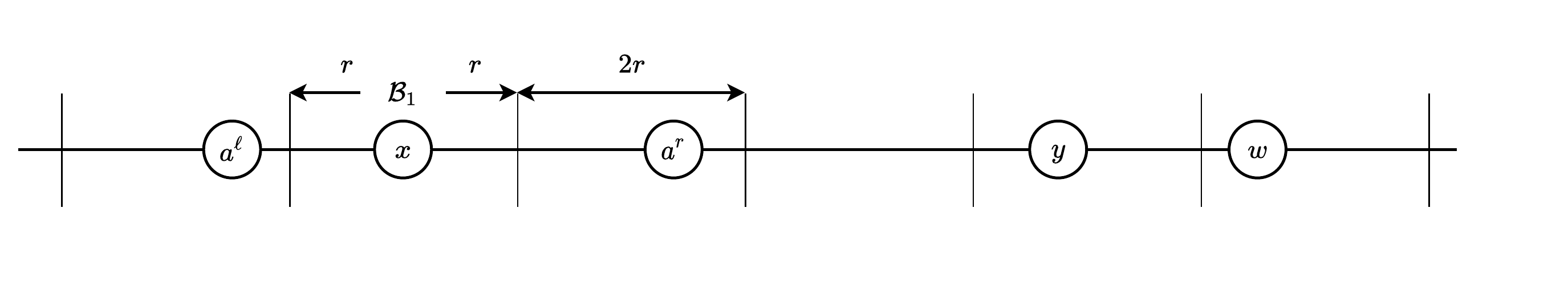}
    \caption{An illustration of our argument for \Cref{theorem:STV-line}.}
    \label{fig:STV}
\end{figure}

As a result, we have established that $\gamma \leq 1/2$, where recall that $\gamma$ represents the fraction of agents in $\mathcal{B}_1$, and the theorem follows by \Cref{lemma:wave}.
\end{proof}

\subsection{Main Result} Moving on to the main result of this section, we will prove the following theorem:

\begin{theorem}
    \label{theorem:STV-doubling}
    If $d$ is the doubling dimension of $\mathcal{M}$, then the distortion of $\stv$ is $O(d \log \log m)$.
\end{theorem}

Before we proceed with the proof of this theorem, let us first present an analysis for general metric spaces. Our argument will uncover the same bound $O(\log m)$ in terms of distortion, as in \cite{DBLP:conf/aaai/SkowronE17}, but it is considerably simpler, and it will also be used in the proof of \Cref{theorem:STV-doubling}.

\begin{theorem}
    \label{theorem:STV-general_metrics}
    The distortion of $\stv$ is $O(\log m)$.
\end{theorem}

\begin{proof}
Let $w \in C$ be the winner of $\stv$ under some sequence of eliminations, and $x \in C$ be the candidate who minimizes the social cost. Moreover, let $r := \dist(x,w)/(4\mathcal{H}_m +2)$, where $\mathcal{H}_m$ denotes the $m$-th harmonic number. If $V_1$ represents the subset of voters in $\mathcal{B}(x, r)$ and $\gamma := |V_1|/n$, we will show that $\gamma \leq 1/2$. Then, our claim will follow from \Cref{lemma:wave}.

For the sake of contradiction, let us assume that $\gamma > 1/2$. We will establish that no voter in $V_1$ will support candidate $w$ at any stage of $\stv$, which is an obvious contradiction since $\gamma > 1/2$ and $w$ was assumed to be the winner. In particular, let $\overline{D}^{(t)}$ be defined as follows:

\begin{equation}
    \label{eq:average_D}
    \overline{D}^{(t)} := \frac{1}{\gamma n} \sum_{i \in V_1} \dist(x, \topp(i; t)),
\end{equation}
where $\topp(i; t)$ represents the most preferred (active) candidate for voter $i$ after round $t = 1, \dots, m-1$. The quantity $\overline{D}^{(0)}$ is also defined as in \Cref{eq:average_D}, assuming that $\topp(i; 0) := \topp(i)$. Thus, observe that the triangle inequality yields that

\begin{equation}
    \overline{D}^{(0)} \leq \frac{1}{\gamma n} \sum_{i \in V_1} ( \dist(x, i) + \dist(i, \topp(i)) < 2r.
\end{equation}

Moreover, we claim that if a voter $i$ supports a candidate $a$ at round $t$, then $\dist(x, a) \leq \overline{D}^{(t-1)} + 2r$. Indeed, we will show the following: If two voters $i, j$ in $\mathcal{B}(x, r)$ support two candidates $a, b$ respectively, then it follows that $\dist(x, a) \leq \dist(x, b) + 2r$. In particular, successive applications of the triangle inequality yield that $\dist(i, a) \leq \dist(i, b) \leq \dist(x, i) + \dist(x, b) \leq r + \dist(x, b)$, while $\dist(i, a) \geq \dist(x, a) - \dist(x, i) \geq \dist(x, a) - r$, in turn implying that 

\begin{equation}
    \label{eq:recast}
    \dist(x, a) \leq \dist(x, b) + 2r.
\end{equation}

By symmetry, it also follows that $\dist(x, b) \leq \dist(x, a) + 2r$. Next, we will (inductively) establish that the quantity $\overline{D}^{(m-2)}$ is strictly less than $4 \mathcal{H}_m \times r$. Indeed, first note that under the invariance $\overline{D}^{(t)} < 4 \mathcal{H}_m \times r$ the voters in $\mathcal{B}(x, r)$ support at least two distinct candidates; otherwise, the unique supported candidate $a$ would prevail since $\gamma > 1/2$, which in turn is a contradiction given that $\dist(x, w) > \overline{D}^{(t)} + 2r \implies a \neq w$. Next, observe that at round $t = 1, \dots, m-1$ at most $n/(m-t+1)$ agents alter their support. This follows since there are exactly $m-t+1$ candidates, while $\stv$ eliminates the one who enjoys the least amount of support. Moreover, all the agents who recast their support will end up coalescing with a candidate whose distance from $x$ increases by at most a $2r$ additive factor---compared to the previously supported alternative; this follows by the bound of \Cref{eq:recast}, and the fact that there is indeed another candidate supported by voters in $\mathcal{B}(x, r)$ in round $t$ (as we previously argued). As a result, we have shown the following recursive structure:

\begin{equation}
    \overline{D}^{(t)} \leq \overline{D}^{(t-1)} + \frac{1}{\gamma} \frac{2r}{m-t + 1} \leq \overline{D}^{(t-1)} + \frac{4r}{m - t + 1},
\end{equation}
for all $t = 1, \dots, m-2$. This verifies the assertion that $\overline{D}^{(m-2)} < 4 \mathcal{H}_m \times r$. Thus, in the ultimate round of $\stv$ more than half the voters support a candidate $a$ for which $\dist(a, x) < 2r + 4 r \times \mathcal{H}_m$, which is a contradiction since $\dist(x, w) = 2r \times (2\mathcal{H}_m + 1)$.
\end{proof}

Next, we provide the proof of \Cref{theorem:STV-doubling}. In particular, the main technical challenge of the analysis lies in maintaining the appropriate invariance during $\stv$. We address this with a simple trick, essentially identifying a subset of the domain with a sufficient degree of regularity. We should also note that the second part of the proof makes use of the technique devised by Skowron and Elkind \cite{DBLP:conf/aaai/SkowronE17}.

\begin{proof}[Proof of \Cref{theorem:STV-doubling}]
As before, let $w \in C$ be the winner of $\stv$ under some sequence of eliminations, and $x \in C$ be the candidate who minimizes the social cost. Moreover, let $r := \dist(x, w)/(4h + 7)$, where $h$ is defined as $h := 1 + \lceil \log_2 ( 6 \lambda^{\log \mathcal{H}_m + 1} ) \rceil = \Theta(d \log \log m)$. If $\gamma$ represents the fraction of the voters in $\mathcal{B}(x, r)$, we will establish that $\gamma \leq 2/3$.

For the sake of contradiction, let us assume that $\gamma > 2/3$. Our argument will characterize the propagation of the support of the voters in $\mathcal{B}(x, r)$. In particular, we proceed in the following two phases:

\textbf{Phase I.} Our high-level strategy is to essentially employ the argument in the proof of \Cref{theorem:STV-general_metrics}, but not for the entire set of voters in $\mathcal{B}(x, r)$. Instead, we will establish the existence of a set with a helpful invariance, which still contains most of the voters. More precisely, we first consider a covering $\{ \mathcal{B}(z_j, r_j) \}_{j=1}^{\mu}$ of the ball $\mathcal{B}(x, r)$, where the radius of every ball is at most $\epsilon \times r$ for some parameter $\epsilon \in (0,1)$. The balls that do not contain any voter may be discarded for the following argument. We let $\mathcal{S}^{(0)}$ be the union of these balls. We know from \Cref{lemma:covering} that $\mu = \mu(\epsilon; \lambda) \leq \lambda^{\log(1/\epsilon) + 1}$. For Phase I we assume that more than $M$ candidates remain active in $\stv$, where $M := 6 \mu$, while $\epsilon := 1/\mathcal{H}_m$.

Let us consider a round $t = 1, \dots, m - M$ of $\stv$. In particular, let $a \in C$ be the candidate who is eliminated at round $t$. Observe that if $a$ is not supported by any voter residing in $\mathcal{B}(x, r)$, the support of these agents remains invariant under round $t$. Thus, let us focus on the contrary case. Specifically, if there exists a ball in the covering which contains exclusively supporters of candidate $a$, we shall remove every such ball from the current covering, updating analogously the set $\mathcal{S}^{(t)}$. Given that we are at round $t$, we can infer that the number of such supporters is at most $n/(m - t + 1) < n/M$. Thus, since we can only remove $\mu$ balls from the initial covering, it follows that the set $\mathcal{S} := \mathcal{S}^{(t)}$ with $t = m - M$ contains strictly more than $2n/3 - n \mu/M = n/2$ voters.

Next, we will argue about the propagation of the support for the voters in $\mathcal{S}$ during the first $m - M$ rounds of $\stv$. By construction of the set $\mathcal{S}$, we have guaranteed the following invariance: Whenever a candidate $a$ supported by voters in $\mathcal{S}$ gets eliminated, every supporter of $a$ from $\mathcal{S}$ lies within a ball of radius at most $\epsilon$ with agents championing a different candidate. Now, let us define $\overline{D}^{(t)}$ as follows:

\begin{equation}
    \overline{D}^{(t)} := \frac{1}{\gamma' n} \sum_{i \in \mathcal{S}} \dist(i, \topp(i; t)),
\end{equation}
where $\gamma'$ represents the fraction of the voters residing in $\mathcal{S}$. Note that $\overline{D}^{(t)}$ is defined slightly differently than in the proof of \Cref{theorem:STV-general_metrics}. Consider two voters $i, j$ supporting two candidates $a, b$ respectively. We will show that $\dist(i, b) \leq \dist(i, a) + 2 \dist(i, j)$, and similarly, $\dist(j, a) \leq \dist(j, b) + 2 \dist(i, j)$. Indeed, successive applications of the triangle inequality imply that 

\begin{align*}
    \dist(i, b) &\leq \dist(i, j) + \dist(j, b) \\
    &\leq \dist(i, j) + \dist(j, a) \\
    &\leq \dist(i, j) + \dist(j, i) + \dist(i, a) \\
    &= \dist(i, a) + 2\dist(i, j).
\end{align*}

Thus, if the voters $i$ and $j$ happen to reside within a ball of radius at most $\epsilon$, we can infer that $\dist(i, b) \leq \dist(i, a) + 4 \epsilon$. As a result, by the recursive argument of \Cref{theorem:STV-general_metrics} we can conclude that 

\begin{equation}
    \overline{D}^{(t)} \leq \overline{D}^{(t-1)} + \frac{1}{\gamma'} \frac{4 \epsilon r}{m - t + 1} \leq \overline{D}^{(t-1)} + \frac{8 \epsilon r}{m - t +1},
\end{equation}
in turn implying that 

\begin{equation}
    \overline{D}^{(m - M)} \leq 8 (\epsilon r) \mathcal{H}_m.
\end{equation}

Consequently, we have essentially shown that the propagation of the support is ``decelerated'' by a factor of $\epsilon$. In particular, for $\epsilon = 1/\mathcal{H}_m$ this implies that during the first phase the agents in $\mathcal{S}$ support candidates within $O(1) \times r$ distance from candidate $x$.

\textbf{Phase II.} At the beginning of the second phase there are $M$ remaining candidates. Let us denote with $\mathcal{B}_j := \mathcal{B}(x, (2 j - 1) \times r)$. In this phase we will argue about the entire set of voters in $\mathcal{B}(x, r)$. Let $m_1 \leq M$ be the number of candidates supported by voters in $\mathcal{B}(x, r)$ at the beginning of the second phase. Our previous argument implies that every such candidate will reside in $\mathcal{B}_7$; this follows directly by applying the triangle inequality. Let us denote with $m_j$ the number of candidates residing outside $\mathcal{B}_{4 + 2j}$ for $j \geq 2$ at the round the last candidate from $\mathcal{B}_{3 + 2j}$ gets eliminated. 

By the pigeonhole principle, we can infer that there exists a candidate $a$ in $\mathcal{B}_7$ who enjoys the support of at least $\gamma n/m_1$ voters. Moreover, observe that the triangle inequality implies that no voter will support a candidate outside $\mathcal{B}_{8}$ as long as candidate $a$ remains active. Thus, at the round $a$ gets eliminated we can deduce that $ (1 - \gamma) n/m_2 \geq \gamma n/m_1 \iff m_2 \leq m_1 \times (1-\gamma)/\gamma$, where we used that the number of candidates in every subset can only decrease during $\stv$. Inductively, we can infer that

\begin{equation}
    m_h \leq \left( \frac{1 - \gamma}{\gamma} \right)^{h-1} m_1 < \left( \frac{1}{2} \right)^{h-1} M \leq 1,
\end{equation}
for $h = \lceil \log_2 M \rceil + 1$, where we used that $\gamma > 2/3$. This implies that the winner of $\stv$ should lie within $B_{4 + 2h}$, i.e. $\dist(x,w)/r < 4h + 7$, which is a contradiction since $\dist(x,w) = (4h + 7) \times r$. Thus, the theorem follows directly from \Cref{lemma:wave}.
\end{proof}

\subsection{The Lower Bound}

In this subsection we refine the $\Omega(\sqrt{\log m})$ lower bound of Skowron and Elkind \cite[Theorem 4]{DBLP:conf/aaai/SkowronE17} based on the doubling dimension of the submetric induced by the set of candidates $\mathcal{M}_C$. In particular, we will establish the following theorem:

\begin{theorem}[Lower Bound for $\stv$]
    \label{theorem:lower_bound-STV}
For any $\lambda \geq 2$ there exists a metric space induced by the set of candidates $(\mathcal{M}_C, \dist)$, with $d = \Theta(\log \lambda)$ being the doubling dimension of $\mathcal{M}_C$, and a voting profile such that the distortion of $\stv$ is $\Omega(\sqrt{d})$.
\end{theorem}

For the proof, we consider first a tree $\mathcal{T}$ with $\lambda$ number of leaves; it will be assumed that $\lambda$ is such that $\lambda = a_i$ for some $i \in \mathbb{N}$, where $\{a_i\}_{i \in \mathbb{N}}$ is a sequence such that $a_1 = 2$ and $a_{i+1} = 2(a_i + 1)$ for $i \geq 1$. Notice that if this is not the case we can always select the maximal $\lambda'$ smaller than $\lambda$ that satisfies this property; given that $\lambda' = \Theta(\lambda)$ this would not affect the conclusion (up to constant factors). Then, the next (or first) \emph{layer} will consist of nodes which are parents of leaves, and in particular, every node in layer $1$ will be parent to exactly $2$ (mutually distinct) leaves, and we will say that the \emph{branching factor} is $b_1 := 2$. This construction is continued iteratively until we reach the root, with the branching factor of layer $i > 1$ satisfying $b_{i+1} = 2 (b_i + 1)$; the first two layers of this construction are illustrated in \Cref{fig:growing_tree}. Observe that by construction the branching factor increases exponentially fast. Moreover, the number of nodes in the $i$-th layer is $m_{i} = m_{i-1}/b_{i}$, with $m_0 := \lambda$. Now let $h$ be the height of the induced tree. We can infer that 

\begin{equation}
    m_{h} = \frac{m_{h-1}}{b_h} = \frac{\lambda}{\prod_{i=1}^h b_i} \geq \frac{\lambda}{\prod_{i=1}^h 4^i} = \frac{\lambda}{4^{\sum_{i=1}^h i}} = \frac{\lambda}{2^{h (h+1)}},
\end{equation}
where we used that $b_{i+1} \leq 4 b_i$. Thus, since $m_h = 1$, it follows that $h = \Omega(\sqrt{\log_2 \lambda})$. Finally, we incorporate a node which is connected via edges to all the leaves. Then, the distance between two nodes is defined as the length of the shortest path in the induced \emph{unweighted} graph. The metric space we introduced will be henceforth represented as $(\mathcal{M}_C, \dist)$.

\begin{claim}
    \label{claim:lower_bound-dimension}
The doubling dimension of the metric space $(\mathcal{M}_C, \dist)$ is $\Theta(\log \lambda)$.
\end{claim}

\begin{proof}
It is easy to see that the doubling constant of the metric space $(\mathcal{M}_C, \dist)$ is at least $\lambda$. Thus, the claim follows since $|\mathcal{M}_C| \leq 2 \lambda$ and $\dim(\mathcal{M}_C) \leq \log_2 |\mathcal{M}_C|$.
\end{proof}

\begin{figure}[!ht]
    \centering
    \includegraphics[scale=0.51]{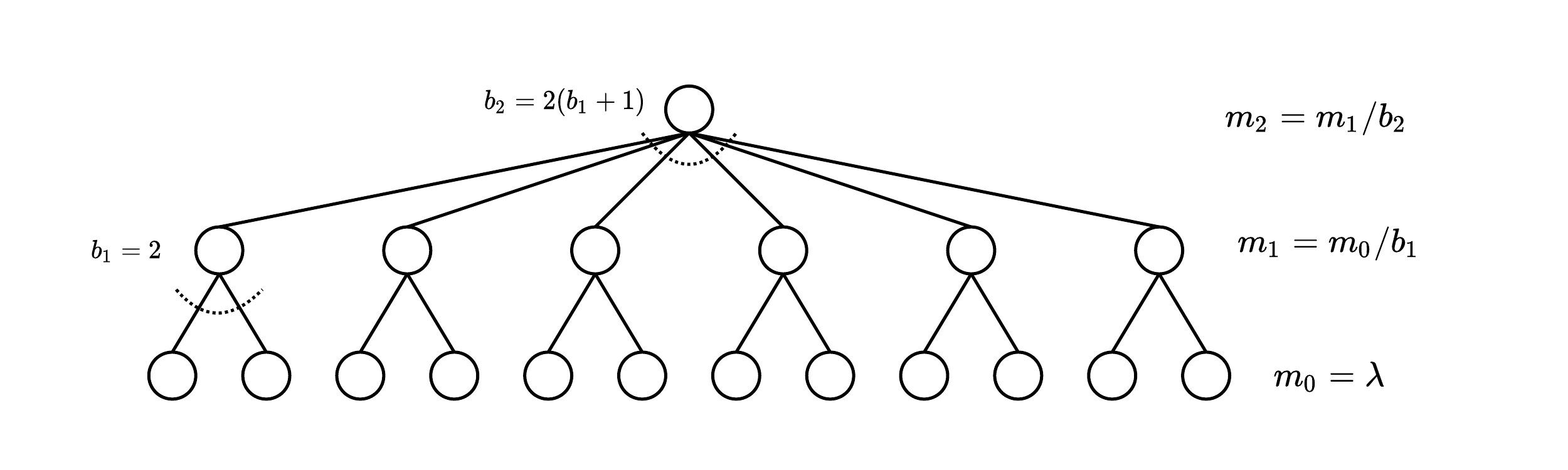}
    \caption{A $2$-layer instance of the tree $\mathcal{T}$ employed for the lower bound in \Cref{theorem:lower_bound-STV}.}
    \label{fig:growing_tree}
\end{figure}

\textit{The Voting Instance.} We assume that voters and candidates are mapped to points on the metric space $(\mathcal{M}_C, \dist)$. In particular, for every point $x \in \mathcal{T}$ we assign a (distinct) candidate, while every remaining candidate will be allocated to the point connected to all the leaves. In this way, $(\mathcal{M}_C, \dist)$ is indeed the metric space induced by the set of candidates. We will let $x \in C$ be a candidate located to the point connected to all the leaves, and $w \in C$ be the candidate at the root of $\mathcal{T}$. Moreover, for a layer $i$ of $\mathcal{T}$ we place $\nu_i$ number of voters at each point of the layer, such that $\nu_0 = 1$ and $\nu_{i+1} = (b_{i} + 1) \nu_i$ for all $i \geq 0$ (here we tacitly assume that $b_0 = 0$); no voters are collocated with candidate $x$.

\begin{claim}
    \label{claim:lower_bound-winner}
There exists an elimination sequence such that $w \in C$ is the winner of $\stv$.
\end{claim}

\begin{proof}
Given that $\nu_{i+1} = (b_i + 1) \nu_i$ we can inductively infer that there exists an elimination sequence such that every candidate in layer $i$ is eliminated before any candidate in the layer $i+1$, while every candidate collocated with $x$ will be eliminated before any candidate in $\mathcal{T}$. This concludes the proof.
\end{proof}

\begin{claim}
    \label{claim:lower_bound-approx}
    $\socialcost(w)/\socialcost(x)  = \Omega(h)$.
\end{claim}

\begin{proof}
Let $n_i$ be the total number of agents residing in the $i$-th layer of the tree $\mathcal{T}$. We will show that $n_{i+1} = n_i/2$, for all $i = 0, 1, \dots, h-1$. Indeed, for $i \geq 0$ it follows that $n_{i+1} = \nu_{i+1} m_{i+1} = (b_{i} + 1) \nu_i m_{i+1} = n_i (b_i + 1)/b_{i+1} = n_i/2$, where we used that $m_i = m_{i-1}/b_i$ and $b_{i+1} = 2 (b_i + 1)$. As a result, it follows that $\socialcost(w) \geq n_0 \times h$, while it is easy to show that $\socialcost(x) \leq 4 n_0$, concluding the proof.
\end{proof}

\begin{proof}[Proof of \Cref{theorem:lower_bound-STV}]
The lower bound follows by directly applying and combining \Cref{claim:lower_bound-dimension}, \Cref{claim:lower_bound-winner}, and \Cref{claim:lower_bound-approx}.
\end{proof}

Notice that this theorem implies as a special case the $\Omega(\sqrt{\log m})$ lower bound for general metrics, which only applies when the metric $(\mathcal{M}_C, \dist)$ is \emph{near-uniform}. 

\begin{remark}
It is not difficult to show that the distortion of $\stv$ is always $O(\Delta)$, where $\Delta$ represents the \emph{aspect ratio} of $\mathcal{M}_C$---the ratio between the largest (pairwise) distance to the smallest distance in $\mathcal{M}_C$. In fact, this bound is tight---up to constant factors---for certain instances, as implied by the construction in \Cref{theorem:lower_bound-STV}.
\end{remark}

\section{Coordination Dynamics}
\label{section:coordination}

In this section we explore the degree to which natural and distributed learning dynamics can converge to social choices with near-optimal distortion. We should point out that there is a concrete connection between such considerations and the results of the previous section, which will be revealed in detail very shortly. First, let us commence with the following observation: 

\begin{observation}
    \label{observation:large_coalitions}
Consider a voting instance under a metric space so that some candidate $a \in C$ has distortion at least $\mathfrak{D}$.\footnote{That is, $\socialcost(a)/\min_{x \in C} \socialcost(x) \geq \mathfrak{D}$.} Then, there exists a candidate $x \neq a$ and subset $W \subseteq V$ such that 
\begin{enumerate}
    \item Every agent in $W$ strictly prefers $x$ to $a$;
    \item $|W|/n \geq 1 - 2/(\mathfrak{D} + 1)$. 
\end{enumerate}
\end{observation}

This statement essentially tells us that candidates with large distortion are inherently unstable, in the sense that there will exist a large ``coalition'' of voters that strictly prefer a different outcome. Interestingly, this observation implies a connection between (metric) distortion and the notion of \emph{core} in cooperative game theory. To be more precise, we will say that a set of coalitions $\mathcal{W}$ is $\alpha$-\emph{large}, with $\alpha \in [0, 1]$, if it contains every coalition $W \subseteq V$ such that $|W|/n \geq \alpha$; a candidate $a$ is said to be in the core if there does \emph{not} exist a coalition $W \in \mathcal{W}$ such that every agent in $W$ (strictly) prefers a different alternative.\footnote{Considering only ``large'' coalitions is standard in the literature; cf. \cite{10.5555/3033138}.} In this context, the following proposition follows directly from \Cref{observation:large_coalitions}:

\begin{proposition}
    \label{proposition:core}
Consider a voting instance under a metric space so that some candidate $a \in C$ has distortion at least $\mathfrak{D}$. Then, candidate $a$ cannot be in the core with respect to an $\alpha$-large set of coalitions, as long as $\alpha \leq 1 - 2/(\mathfrak{D} + 1)$.
\end{proposition}

As a result, it is interesting to study the strategic behavior and the potential coordination dynamics that may arise in the face of an inefficient voting system. 

\subsection{The Model} We consider the following abstract model: For some given voting system, agents are called upon to cast their votes for a series of $T$ days or rounds, where $T$ is sufficiently large. After the end of each day, voters are informed about the results of the round, and the winner is determined based on the results of the \emph{ultimate} day. This is essentially an iterative implementation of a given voting rule, in place of the \emph{one-shot} execution typically considered, and it is introduced to take into account external information typically accumulated before the actual voting (e.g. through polls). For concreteness, we will assume that the voting rule employed in each day is simply the $\textsc{Plurality}$ mechanism, not least due to its popularity both in theory and in practice.

Before we describe and analyze natural dynamics in this model, let us first note that if all the voters engage truthfully throughout this game, the victor will coincide with the plurality winner, and as we know there are instances for which this candidate may have $\Omega(m)$ distortion. As a result, \Cref{observation:large_coalitions} implies that there will be a large coalition with a $1 - \Theta(1/m)$ fraction of the voters that strictly prefer a different outcome. Indeed, the lower bound of $\textsc{Plurality}$ is built upon $m-1$ clusters of voters formed arbitrarily close, while a different extreme party with roughly the same plurality score could eventually prevail. However, the access to additional information renders this scenario rather unrealistic given that we expect some type of adaptation or coordination mechanism from the agents. 

\subsection{A Greedy Approach} Let us denote with $n_a^{(t)}$ the plurality score of candidate $a$ at round $t \in [T]$. A particularly natural approach for an agent to engage in this scenario consists of maintaining a time-varying parameter $\theta^{(t)}$, which will essentially serve as the ``temperature''. Then, at some round $t > 1$ agent $i$ will support the candidate $b$ for which $b \succeq_i a$ for all $a, b \in C^{(t)}$, where $C^{(t)} := \{ a \in C : n_a^{(t-1)} \geq \theta^{(t)} \}$.\footnote{The definition of the set $C^{(t)}$ for $t > 1$ is subject to $|C^{(t)}| \geq |C^{(t-1)}| - 1$, i.e. agents never disregard more than $1$ candidate in the course of a single round; in the contrary case the guarantee we state for the dynamics does not hold due to some pathological instances.} That is, agents only consider candidates who exceeded some level of support during the previous day. Then, the temperature parameter is updated accordingly, for example with some small constant increment $\theta^{(t+1)} := \theta^{(t)} + \epsilon$, for some $\epsilon > 0$. In this context, observe that for a sufficiently small $\epsilon$ these dynamics will converge to an $\stv$ winner (based on the parallel universe model). This implies that the greedy tactics already offer an exponential improvement---in terms of the utilitarian efficiency---compared to the truthful dynamics. Nevertheless, the lower bound for $\stv$ (\Cref{theorem:lower_bound-STV}) suggests that we have to design a more careful adaptive rule in order to attain $O(1)$ distortion. 

\subsection{Exploration/Exploitation} The inefficiency of the previous approach---and subsequently of $\stv$---stems from the greedy nature of the iterative process: Agents may choose to dismiss candidates prematurely. For example, this becomes immediately apparent by inspecting the elimination pattern in the lower bound of \Cref{theorem:lower_bound-STV}. In light of this, the remedy we propose---and what arguably occurs in many practical scenarios---is an exploration phase. In particular, voters initially do not possess any information about the preferences of the rest of the population. Thus, they may attempt to explore several alternatives in order to evaluate the viability of each candidate; while doing so, agents will endeavor to somehow indicate or favor their own preferences. After the exploration phase, agents will leverage the information they have learnt to adapt their support. More concretely, we will consider the following dynamics:

\begin{enumerate}
    \item Exploration phase: In each round $t \in [m]$ every agent $i$ maintains a list $\mathcal{L}_i^{(t)}$, initialized as $\mathcal{L}_i^{(1)} := \emptyset$. If $C_i^{(t)} := C \setminus \mathcal{L}_i^{(t)}$, then at round $t$ an agent $i$ shall vote for the candidate $a \in C_i^{(t)}$ such that $a \succeq_i b$ for all $b \in C_i^{(t)}$. Then, agent $i$ updates her list accordingly: $\mathcal{L}_i^{(t+1)} := \mathcal{L}_i^{(t)} \cup \{a \}$;
    \item Exploitation phase: Every agent supports the first candidate\footnote{For simplicity, it is assumed that in case multiple such agents exist we posit some arbitrary but \emph{common} among all agents tie-breaking mechanism.} within her list that managed to accumulate---over all prior rounds---at least $n/2$ votes.
\end{enumerate}

In a sense, voters try to balance between voting for their most-preferred candidates and having an impact on the final result. We shall refer to this iterative process as $\coordination$ dynamics.

\begin{theorem}
    $\coordination$ dynamics lead to a candidate with distortion at most $11$.
\end{theorem}

\begin{proof}
Let $w \in C$ be the winner under $\coordination$ dynamics, and $x \in C$ be the candidate who minimizes the social cost. If $r := \dist(x,w)/5$, we consider the sequence of balls $ \{ \mathcal{B}_i \}_{i=1}^3$ such that $\mathcal{B}_i := \mathcal{B}(x, (2i-1)r)$ for $i=1,2,3$. If $\gamma$ is the fraction of the voters in $\mathcal{B}_1$, we will argue that $\gamma \leq 1/2$.

For the sake of contradiction, let us assume that $\gamma > 1/2$. Let $t$ be the first round for which a voter in $\mathcal{B}_1$ supports a candidate outside $\mathcal{B}_3$. Then, it follows by the triangle inequality that the list of this voter just after round $t-1$ included all the candidates in $\mathcal{B}_2$. This in turn implies that by round $t-1$ every agent in $\mathcal{B}_1$ had already voted for all candidates in $\mathcal{B}_1$. Given that $\gamma > 1/2$, we can conclude that no agent from $\mathcal{B}_1$ voted for $w$ during the exploitation phase. 

Now let us consider the first round for which some candidate $a \in C$ accumulated at least $n/2$ votes, which clearly happens during the exploration phase. Then, at the exact same round at least $n/2$ agents have $a$ in their list; this follows since agents vote for different candidates during the exploration phase, and a candidate is always included in the list once voted for. As a result, our tie-breaking assumption implies that there will be a candidate with the support of at least $n/2$ agents during the exploitation phase. But our previous argument shows that this candidate cannot be $w$, which is an obvious contradiction. As a result, we have shown that $\gamma \leq 1/2$, and the theorem follows by \Cref{lemma:wave}.
\end{proof}

Before we conclude this section, let us briefly mention some intriguing open problems related to our results. Specifically, we have attempted to argue that candidates with small distortion may arise through natural learning rules. This was motivated in part by \Cref{observation:large_coalitions}, which implies the instability of outcomes with large distortion. However, the converse of this statement is not quite true: Although there always exists a candidate with distortion at most $3$ \cite{DBLP:conf/focs/Gkatzelis0020}, there might be a subset with at least half of the voters that strictly prefer a different outcome (a.k.a. Condorcet's paradox). Still, there might be an appropriate notion of stability which ensures that near-optimal candidates are in some sense stable. In spirit, this is very much pertinent to the main result of Gkatzelis et al. \cite{DBLP:conf/focs/Gkatzelis0020} concerning the existence of an \emph{undominated} candidate, leading to the following question:

\begin{question}
    Are there deterministic and distributed learning rules which converge to a candidate with distortion $3$?
\end{question}

\section{Intrinsic Robustness of Plurality Matching}
\label{section:robustness}

Gkatzelis et al. \cite{DBLP:conf/focs/Gkatzelis0020} introduced the (deterministic) $\pluralitymatching$ mechanism, and they showed that it always incurs distortion at most $3$ under metric preferences. Nonetheless, it is natural to ask how it performs in more refined, as well as in more general spaces. It should be noted that Gkatzelis et al. \cite{DBLP:conf/focs/Gkatzelis0020} established the robustness of $\pluralitymatching$ under different objective functions (measuring the social cost); namely, they showed that the same distortion bound can be achieved for the more stringent \emph{fairness ratio} of Goel et al. \cite{10.1145/3033274.3085138}. In this section we extend the robustness of $\pluralitymatching$ along two regimes.

\subsection{Ultra-Metrics} First, we study the power of $\pluralitymatching$ under \emph{ultra-metric spaces}; in particular, recall the following definition:

\begin{definition}
    \label{definition:ultra-metric}
An ultra-metric on a set $\mathcal{M}$ is a function $\dist : \mathcal{M} \times \mathcal{M} \mapsto \mathbb{R}$ such that $\forall x, y, z$,

\begin{enumerate}
    \item $\dist(x,y) = 0$ if and only if $x = y$ (identity of indiscernibles);
    \item $\dist(x,y) = \dist(y,x)$ (symmetry);
    \item $\dist(x, z) \leq \max \{\dist(x,y), \dist(y, z) \}$ (ultra-metric inequality).
\end{enumerate}
\end{definition}

Notice that these axioms also imply that $\dist(x,y) \geq 0, \forall x, y \in \mathcal{M}$. We will say that an ultra-metric space is an ordered pair $(\mathcal{M}, \dist)$ consisting of a set $\mathcal{M}$ along with an ultra-metric $\dist$ on $\mathcal{M}$. Naturally, every ultra-metric is also a metric since $\max \{ \dist(x,y), \dist(y,z) \} \leq \dist(x,y) + \dist(y,z)$, but the converse is not necessarily true. Perhaps the simplest conceivable ultra-metric is the \emph{discrete metric}, which is defined on a set $\mathcal{M}$ as follows:

$$
\dist(x,y)=
\begin{cases}
1 \quad \text{if $x \neq y$;} \\
0 \quad \text{if $x=y$},
\end{cases}
$$
where $x, y \in \mathcal{M}$. As we explained in our introduction, we study this setting mostly driven by the fundamental bottleneck variant in facility location games. Specifically, if the cost of a path corresponds to the maximum-weight edge in the path (posit non-negative weights), and the distance between two nodes in the graph is the minimum-cost path among all possible paths, it is well-known that these (so-called \emph{minimax}) distances satisfy the ultra-metric inequality of \Cref{definition:ultra-metric}. In \Cref{tab:ultra} we summarize some lower bounds for well-studied mechanisms; they mostly follow directly from the techniques of Anshelevich et al. \cite{DBLP:conf/aaai/AnshelevichBP15}, and thus we omit their proof. 

\begin{table}[!ht]
    \centering
    \begin{tabular}{||c c||} 
 \hline
 Mechanism & Lower Bound \\ [0.5ex] 
 \hline\hline
 Plurality \& Borda & $m$ \\ 
 \hline
 $k$-top & $\Omega(m/k)$ \\
 \hline
 Approval \& Veto & $n$ \\
 \hline
 Any Deterministic & $2$ \\
 \hline
\end{tabular}
\caption{Lower bounds for standard mechanisms under ultra-metric spaces. We use $k$-top to represent any \emph{deterministic} mechanism which elicits only the $k$-top preferences.}
    \label{tab:ultra}
\end{table}

In particular, we first prove a lower bound of $2$ for any deterministic mechanism under ultra-metrics:

\begin{proposition}
    \label{proposition:lower_bound-ultra_metric}
There exists a voting profile for which the distortion of any deterministic mechanism under an ultra-metric space is at least $2$.
\end{proposition}

\begin{proof}
Consider a voting instance with two candidates $a, b$ and $n$ voters, such that the votes between the two candidates are split equally. Assume without any loss of generality that the mechanism eventually selects candidate $b$. We will present an ultra-metric space for which the social cost of $a$ is half than the social cost of $b$. Specifically, consider an unweighted path graph with $3$ nodes endowed with the minimax distance. We assume that candidate $a$ resides in the leftmost node of the graph along with all of the voters who supported $a$; on the other hand, candidate $b$ resides in the rightmost node of the graph, while all of her supporters lie in the intermediate node (see \Cref{fig:ultra_metric}). Then, it follows that $\socialcost(a) = n/2$ whereas $\socialcost(b) = n$, as desired.
\end{proof}

\begin{figure}[!ht]
    \centering
    \includegraphics[scale=0.5]{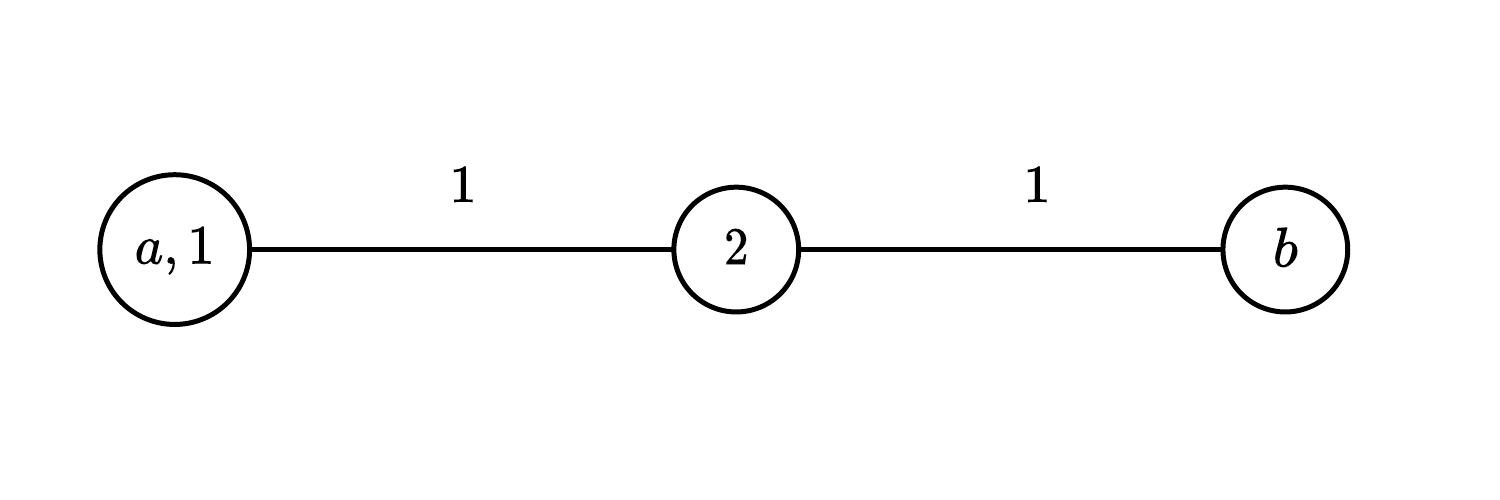}
    \caption{An example of our construction for the proof of \Cref{proposition:lower_bound-ultra_metric} for $n=2$ voters; the argument is similar to the one for metric spaces, but observe that in this case $\dist(1, b) = 1 \neq 2$ since we have considered minimax distances.}
    \label{fig:ultra_metric}
\end{figure}

Importantly, we will show that $\pluralitymatching$ always matches this lower bound. To keep the exposition reasonable self-contained we shall first recall some basic ingredients developed in \cite{DBLP:conf/focs/Gkatzelis0020}.

\begin{definition}[\cite{DBLP:conf/focs/Gkatzelis0020}, Definition $5$]
    For an election $\mathcal{E} = (V, C, \sigma)$ and a candidate $a \in C$, the integral domination graph of candidate $a$ is the bipartite graph $G^{\mathcal{E}}(a) = (V, V, E_a)$, where $(i, j) \in E_a$ if and only if $a \succeq_i \topp(j)$.
\end{definition}

\begin{proposition}[\cite{DBLP:conf/focs/Gkatzelis0020}, Corollary $1$]
    \label{proposition:perfect_matching}
    There exists a candidate $a \in C$ whose integral domination graph $G^{\mathcal{E}}(a)$ admits a perfect matching.
\end{proposition}

We should also note that a candidate whose integral domination graph admits a perfect matching can be identified in strongly polynomial time. In particular, $\pluralitymatching$ always returns such a candidate. These ingredients suffice in order to establish the following:

\begin{theorem}
    \label{theorem:upper_bound-ultra_metric}
$\pluralitymatching$ returns a candidate with distortion at most $2$ under any ultra-metric space.
\end{theorem}

\begin{proof}[Proof of \Cref{theorem:upper_bound-ultra_metric}]
Let $a \in C$ be a candidate whose integral domination graph $G^{\mathcal{E}}(a)$ admits a perfect matching $M : V \mapsto V$ (recall \Cref{proposition:perfect_matching}), such that $a \succeq_i \topp(M(i))$ for all $i \in V$. Then, it follows that 
\begin{align*}
    \socialcost(a) &= \sum_{i \in V} \dist(i, a) \\
           &\leq \sum_{i \in V} \dist(i, \topp(M(i))) & (a \succeq_i \topp(M(i)), \forall i \in V) \\
           &\leq \sum_{i \in V} \dist(i, b) + \dist(b, \topp(M(i))) & (\text{triangle inequality}) \\
           &= \socialcost(b) + \sum_{i \in V} \dist(b, \topp(i))) & (\text{$M$ is a perfect matching}) \\
           &\leq \socialcost(b) + \sum_{i \in V} \max \{ \dist(i, b), \dist(i, \topp(i)) \} & (\text{ultra-metric inequality}) \\
           &= \socialcost(b) + \sum_{i \in V} \dist(i, b) & (\topp(i) \succeq_i b) \\
           &= 2 \socialcost(b).
\end{align*}
Given that the choice of $b$ was arbitrary the theorem follows.
\end{proof}

\begin{remark}
In \Cref{appendix:axiomatic} we also present an axiomatic analysis of $\pluralitymatching$ under a broad class of operators, which (among others) implies that if the social cost is determined by the most remote agent, then $\pluralitymatching$ recovers the optimal utilitarian welfare under ultra-metrics.
\end{remark}

\subsection{Approximate Metrics}

Next, we study the distortion of deterministic mechanisms when the distances \emph{approximately} satisfy the triangle inequality, as formalized in the following definition:

\begin{definition}
For some parameter $\rho \geq 1$, a $\rho$-approximate metric on a set $\mathcal{M}$ is a function $\dist : \mathcal{M} \times \mathcal{M} \mapsto \mathbb{R}$ such that $\forall x, y, z$,

\begin{enumerate}
    \item $\dist(x,y) = 0$ if and only if $x = y$ (identity of indiscernibles);
    \item $\dist(x,y) = \dist(y,x)$ (symmetry);
    \item $\dist(x, z) \leq \rho (\dist(x, y) + \dist(y, z))$ ($\rho$-relaxed triangle inequality).
\end{enumerate}
\end{definition}

Again we point out that these axioms\footnote{It should be noted that $\rho$-approximate metrics cannot be handled with the axiomatic approach given in \Cref{appendix:axiomatic} given that the operator $(x, y) \mapsto \rho ( x + y)$ is \emph{not} associative for $\rho > 1$.} imply that $\dist(x,y) \geq 0, \forall x, y \in \mathcal{M}$. We commence with the following lower bound:

\begin{proposition}
    \label{proposition:lower_bound-approximate}
There exists a voting profile for which the distortion of every deterministic mechanism under a $\rho$-approximate metric space is at least $\rho^2 + \rho + 1$.
\end{proposition}

\begin{proof}[Proof of \Cref{proposition:lower_bound-approximate}]
As usual, consider an instance with $2n$ voters and $2$ candidates $a, b \in C$, such that every candidate obtains exactly half of the votes. Let us assume without any loss of generality that the social choice rule selects candidate $b$. Now consider a $\rho$-approximate metric $\dist(\cdot, \cdot)$ on the set $\mathcal{M} = \{x, y, z, \omega\}$ defined as follows:

\[
\begin{blockarray}{cccccc}
x & y & z & \omega \\
\begin{block}{(cccc)cc}
  0 & 1 & 2\rho & \epsilon & & x \\
  1 & 0 & 1 & \rho + \rho \epsilon & & y \\
  2 \rho & 1 & 0 & \rho^2 + \rho + \rho \epsilon & & z \\
  \epsilon & \rho + \rho \epsilon & \rho^2 + \rho + \rho \epsilon & 0 & & \omega \\
\end{block}
\end{blockarray};
 \]
here we assume that $\epsilon \in (0,1)$. It is a simple exercise to verify that $\dist(\cdot, \cdot)$ indeed satisfies the axioms of a $\rho$-approximate metric. We also assume that $a := x$ and $b := z$; the $n$ supporters of candidate $a$ are located on $\omega$, while the $n$ supporters of candidate $b$ on $y$. Then, it follows that the agents' locations are consistent with their preferences, while the distortion of candidate $b$ reads 

\begin{equation}
    \frac{\socialcost(b)}{\socialcost(a)} = \frac{n d(z, \omega) + n d(y, z)}{n d(x, \omega) + n d(x, y)} = \frac{\rho^2 + \rho + \rho \epsilon + 1}{\epsilon + 1}.
\end{equation}
Taking the supremum of this ratio over $\epsilon \in (0,1)$ concludes the proof.
\end{proof}

To provide some intuition let us assume that $\dist$ corresponds to the squared Euclidean distance.\footnote{Note that Young's inequality implies that $\dist(\cdot, \cdot)$ is a $2$-approximate metric.} If we consider the usual voting scenario wherein the votes are splitted equally among two candidates (see the proof of \Cref{proposition:lower_bound-approximate} and \Cref{fig:squared_eucl}), then it follows that the distortion of candidate $b$ reads

\begin{equation}
    \frac{\socialcost(b)}{\socialcost(a)} = \frac{n \times 1^2 + n \times (2 + \delta)^2}{n \times 1^2 + n \times \delta^2} = \frac{\delta^2 + 4 \delta + 5}{\delta^2 + 1}.
\end{equation}

\begin{figure}[!ht]
    \centering
    \includegraphics[scale=0.6]{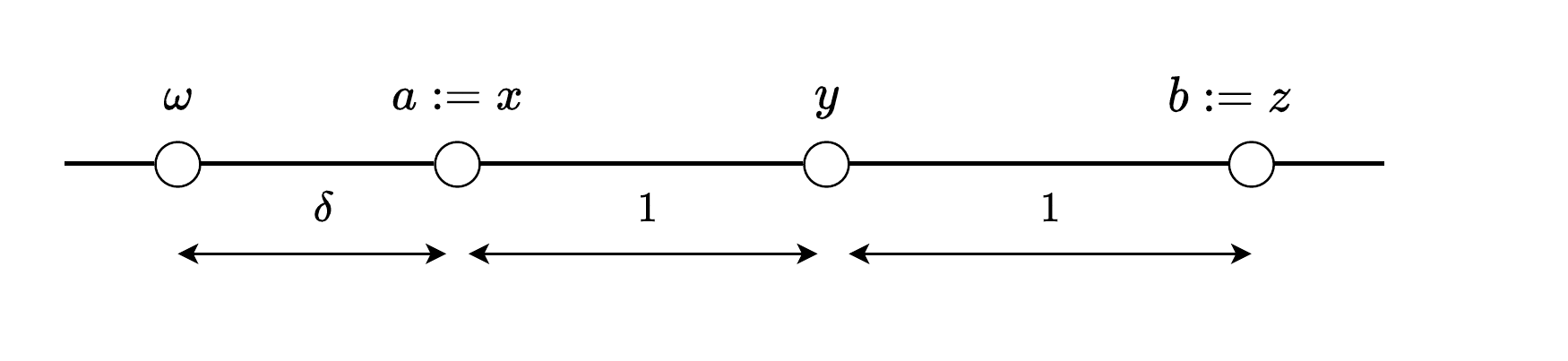}
    \caption{The lower bound for squared Euclidean distances.}
    \label{fig:squared_eucl}
\end{figure}

Interestingly, this ratio increases as $\delta$ goes from $0$ to a sufficiently small constant, implying a notable qualitative difference compared to the standard metric case. In particular, for the squared Euclidean distance it follows that the distortion is lower-bounded by $\sup_\delta \frac{\delta^2 + 4 \delta + 5}{\delta^2 + 1} = (4 + 2 \sqrt{2})/(4 - 2 \sqrt{2}) \cong 5.8284$. For general approximate metrics we can employ the techniques of Gkatzelis et al. \cite{DBLP:conf/focs/Gkatzelis0020} to show the following:

\begin{theorem}
    \label{theorem:upper_bound-approximate}
$\pluralitymatching$ returns a candidate with distortion at most $2\rho^2 + \rho$ under $\rho$-approximate metrics. 
\end{theorem}

As a result, this theorem leaves a gap between the upper bound derived for $\pluralitymatching$ and the lower bound of \Cref{proposition:lower_bound-approximate} when $\rho > 1$. Nonetheless, it should be noted that for $\rho$-approximate metrics there exists an \emph{instance-optimal} and computationally efficient mechanism based on linear programming (cf. \cite{10.1145/3033274.3085138}).

\section*{Acknowledgements} We are indebted to anonymous (AAAI `22) reviewers for providing many comments that helped improve the exposition of this paper. We are also grateful to Benny Moldovanu for pointing out related work on the subject of dimensionality in voting.

\bibliography{paper}
\bibliographystyle{aomplain}

\appendix

\section{An Axiomatic Extension}
\label{appendix:axiomatic}

Here we extend the analysis of Gkatzelis et al. \cite{DBLP:conf/focs/Gkatzelis0020} for a broad class of operators. In particular, we will make the following hypothesis:

\begin{assumption}
    \label{assumption:binary_operator}
    The binary operator $\oplus : \mathbb{R} \times\mathbb{R} \mapsto \mathbb{R}$ satisfies the following for all $x, y, z \in \mathbb{R}$:
    
    \begin{itemize}
        \item $x \oplus y = y \oplus x$ (commutative property);
        \item $(x \oplus y) \oplus z = x \oplus (y \oplus z)$ (associative property);
        \item $x \oplus y \leq x' \oplus y$ for all $x' \geq x$ (monotonicity).
    \end{itemize}
\end{assumption}

Canonical examples of this class of operators include the standard addition $+$ over the reals, as well as the $\max$ operator. We also introduce the following concept:

\begin{definition}
A metric w.r.t. a binary operator $\oplus$ on a set $\mathcal{M}$ is a function $\dist : \mathcal{M} \times \mathcal{M} \mapsto \mathbb{R}$ such that $\forall x, y, z$,

\begin{enumerate}
    \item $\dist(x,y) = 0$ if and only if $x = y$ (identity of indiscernibles);
    \item $\dist(x,y) = \dist(y,x)$ (symmetry);
    \item $\dist(x,z) \leq \dist(x, y) \oplus \dist(y, z)$ (metric inequality w.r.t. $\oplus$).
\end{enumerate}
\end{definition}

Moreover, the operator $\oplus$ will determine the social cost of a candidate; namely, with a slight abuse of notation we let 

\begin{equation}
\socialcost(a) = \sideset{}{}\bigoplus_{i \in V} d(i, a),
\end{equation}
for a candidate $a \in C$; naturally, it is assumed that voters and candidates have been embedded into a metric space w.r.t. the operator $\oplus$. 

\begin{theorem}
    \label{theorem:upper_bound-axiomatic}
$\pluralitymatching$ returns a candidate $a$ such that $\socialcost(a) \leq \socialcost(b) \oplus \socialcost(b) \oplus \socialcost(b), \forall b \in C$, under any metric space w.r.t. an operator $\oplus$ satisfying \Cref{assumption:binary_operator}. In particular, it follows that if $\oplus$ is idempotent then $\socialcost(a) \leq \socialcost(b)$ for all $b \in C$.
\end{theorem}

The proof of this theorem closely follows the argument in \cite{DBLP:conf/focs/Gkatzelis0020}, but we include it for completeness.

\begin{proof}[Proof of \Cref{theorem:upper_bound-axiomatic}]
Let $a \in C$ be a candidate whose integral domination graph $G^{\mathcal{E}}(a)$ admits a perfect matching $M : V \mapsto V$, such that $a \succeq_i \topp(M(i))$ for all $i \in V$. Then, it follows that 

\begin{align*}
    \socialcost(a) &= \sideset{}{}\bigoplus_{i \in V} d(i, a) \\
           &\leq \sideset{}{}\bigoplus_{i \in V} \dist(i, \topp(M(i))) & (a \succeq_i \topp(M(i)), \forall i \in V) \\
           &\leq \sideset{}{}\bigoplus_{i \in V} ( \dist(i, b) \oplus \dist(b, \topp(M(i))) ) & (\text{triangle inequality w.r.t. $\oplus$}) \\
           &= \socialcost(b) \oplus \sideset{}{}\bigoplus_{i \in V} \dist(b, \topp(i))) & (\text{$M$ is a perfect matching}) \\
           &\leq \socialcost \oplus \sideset{}{}\bigoplus_{i \in V} ( \dist(i, b) \oplus \dist(i, \topp(i)) ) & (\text{triangle inequality w.r.t. $\oplus$}) \\
           &\leq \socialcost(b) \oplus \sideset{}{}\bigoplus_{i \in V} ( \dist(i, b) \oplus \dist(i, b) ) & (\text{$\topp(i) \succeq_i b$}) \\
           &= \socialcost(b) \oplus \socialcost \oplus \socialcost(b),
\end{align*}
where we used the properties of the operator (\Cref{assumption:binary_operator}); given that the choice of $b$ was arbitrary we arrive at the desired conclusion.
\end{proof}

One implication of this observation is the following: If the social cost is determined by the agent with the largest distance, i.e. $\socialcost(a) = \max_{i \in V} d(i, a)$, then $\pluralitymatching$ has distortion $1$ in ultra-metric spaces.

\end{document}